%% file: main.tex
\documentclass[sigconf, nonacm, pdfa]{acmart}
\usepackage{amsthm}
\usepackage{algorithmicx}
\usepackage{algorithm}
\usepackage[noend]{algpseudocode}
\algnewcommand\algorithmicinput{\textbf{Parameters:}}
\algnewcommand\Params{\item[\algorithmicinput]}
\newtheorem{theorem}{Theorem}[section]

\usepackage[table]{xcolor}
\usepackage{colortbl}
\usepackage{color}
\usepackage{cleveref}
\usepackage{soul}
\usepackage[a-2b]{pdfx}

\newcommand\vldbdoi{10.14778/3796195.3796215}
\newcommand\vldbpages{1060 - 1073}
\newcommand\vldbvolume{19}
\newcommand\vldbissue{5}
\newcommand\vldbyear{2026}
\newcommand\vldbauthors{\authors}
\newcommand\vldbtitle{\shorttitle} 
\newcommand\vldbavailabilityurl{https://github.com/mitdbg/palimpzest}
\newcommand\vldbpagestyle{empty} 

\begin{document}
\title{Abacus: A Cost-Based Optimizer for Semantic Operator Systems}


\author{Matthew Russo}
\orcid{0009-0005-9685-3976}
\affiliation{%
  \institution{MIT}
}
\email{mdrusso@csail.mit.edu}

\author{Chunwei Liu}
\affiliation{%
  \institution{MIT}
}
\email{chunwei@csail.mit.edu}

\author{Sivaprasad Sudhir}
\affiliation{%
  \institution{MIT}
}
\email{siva@csail.mit.edu}

\author{Gerardo Vitagliano}
\affiliation{%
  \institution{MIT}
}
\email{gerarvit@mit.edu}

\author{Michael Cafarella}
\affiliation{%
  \institution{MIT}
}
\email{michjc@csail.mit.edu}

\author{Tim Kraska}
\affiliation{%
  \institution{MIT}
}
\email{kraska@mit.edu}

\author{Samuel Madden}
\affiliation{%
  \institution{MIT}
}
\email{madden@csail.mit.edu}

\renewcommand{\shortauthors}{Russo et al.}
\newcommand{\system}{{\sc Abacus}}
\renewcommand{\sectionautorefname}{Section}
\newcommand{\tim}[1]{\textcolor{red}{\bf (Tim)~#1}{\typeout{#1}}}
\newcommand{\matt}[1]{\textcolor{blue}{#1}}
\newcommand{\cut}[1]{\textcolor{red}{#1}}
\newcommand{\srm}[1]{\textcolor{purple}{\bf (Sam)~#1}{\typeout{#1}}}
\newcommand{\edit}[1]{\textcolor{blue}{#1}}

\input{sections/abstract}

\maketitle

\pagestyle{\vldbpagestyle}
\begingroup\small\noindent\raggedright\textbf{PVLDB Reference Format:}\\
\vldbauthors. \vldbtitle. PVLDB, \vldbvolume(\vldbissue): \vldbpages, \vldbyear.\\
\href{https://doi.org/\vldbdoi}{doi:\vldbdoi}
\endgroup
\begingroup
\renewcommand\thefootnote{}\footnote{\noindent
This work is licensed under the Creative Commons BY-NC-ND 4.0 International License. Visit \url{https://creativecommons.org/licenses/by-nc-nd/4.0/} to view a copy of this license. For any use beyond those covered by this license, obtain permission by emailing \href{mailto:info@vldb.org}{info@vldb.org}. Copyright is held by the owner/author(s). Publication rights licensed to the VLDB Endowment. \\
\raggedright Proceedings of the VLDB Endowment, Vol. \vldbvolume, No. \vldbissue\ %
ISSN 2150-8097. \\
\href{https://doi.org/\vldbdoi}{doi:\vldbdoi} \\
}\addtocounter{footnote}{-1}\endgroup

\ifdefempty{\vldbavailabilityurl}{}{
\vspace{.3cm}
\begingroup\small\noindent\raggedright\textbf{PVLDB Artifact Availability:}\\
The source code, data, and/or other artifacts have been made available at \url{\vldbavailabilityurl}.
\endgroup
}

\input{sections/intro}
\input{sections/sys_overview}
\input{sections/algorithms}
\input{sections/evaluation}
\input{sections/limitations}
\input{sections/related_work}
\input{sections/conclusion}

\begin{acks}
We are grateful for the support from the DARPA ASKEM Award HR00112220042, the ARPA-H Biomedical Data Fabric project, NSF DBI 2327954, a grant from Liberty Mutual, a Google Research Award, and the Amazon Research Award. Additionally, our work has been supported by contributions from Amazon, Google, and Intel as part of the MIT Data Systems and AI Lab (DSAIL) at MIT, along with NSF IIS 1900933. This research was sponsored by the United States Air Force Research Laboratory and the Department of the Air Force Artificial Intelligence Accelerator and was accomplished under Cooperative Agreement Number FA8750-19-2-1000. The views and conclusions contained in this document are those of the authors and should not be interpreted as representing the official policies, either expressed or implied, of the Department of the Air Force or the U.S. Government. The U.S. Government is authorized to reproduce and distribute reprints for Government purposes notwithstanding any copyright notation herein.
\end{acks}


\bibliographystyle{ACM-Reference-Format}
\bibliography{sample}

\end{document}

%% file: sections/abstract.tex
\begin{abstract}
LLMs enable an exciting new class of data processing applications over large collections of unstructured documents. Several new programming frameworks have enabled developers to build these applications by composing them out of semantic operators: a declarative set of AI-powered data transformations with natural language specifications. These include LLM-powered maps, filters, joins, etc. used for document processing tasks such as information extraction, summarization, and more. While systems of semantic operators have achieved strong performance on benchmarks, they can be difficult to optimize. An optimizer for this setting must determine how to physically implement each semantic operator in a way that optimizes the system globally. Existing optimizers are limited in the number of optimizations they can apply, and most (if not all) cannot optimize system quality, cost, or latency subject to constraint(s) on the other dimensions. In this paper we present \system{}, an extensible, cost-based optimizer which searches for the best implementation of a semantic operator system given a (possibly constrained) optimization objective. \system{} estimates operator performance by leveraging a minimal set of validation examples, prior beliefs about operator performance, and/or an LLM judge. We evaluate \system{} on document processing workloads in the biomedical and legal domains (BioDEX; CUAD) and multi-modal question answering (MMQA). We demonstrate that, on-average, systems optimized by \system{} achieve 6.7\%-39.4\% better quality and are 10.8x cheaper and 3.4x faster than the next best system.
\end{abstract}

%% file: sections/intro.tex
\section{Introduction}
\label{sec:intro}

\begin{figure*}
\includegraphics[width=.9\textwidth]{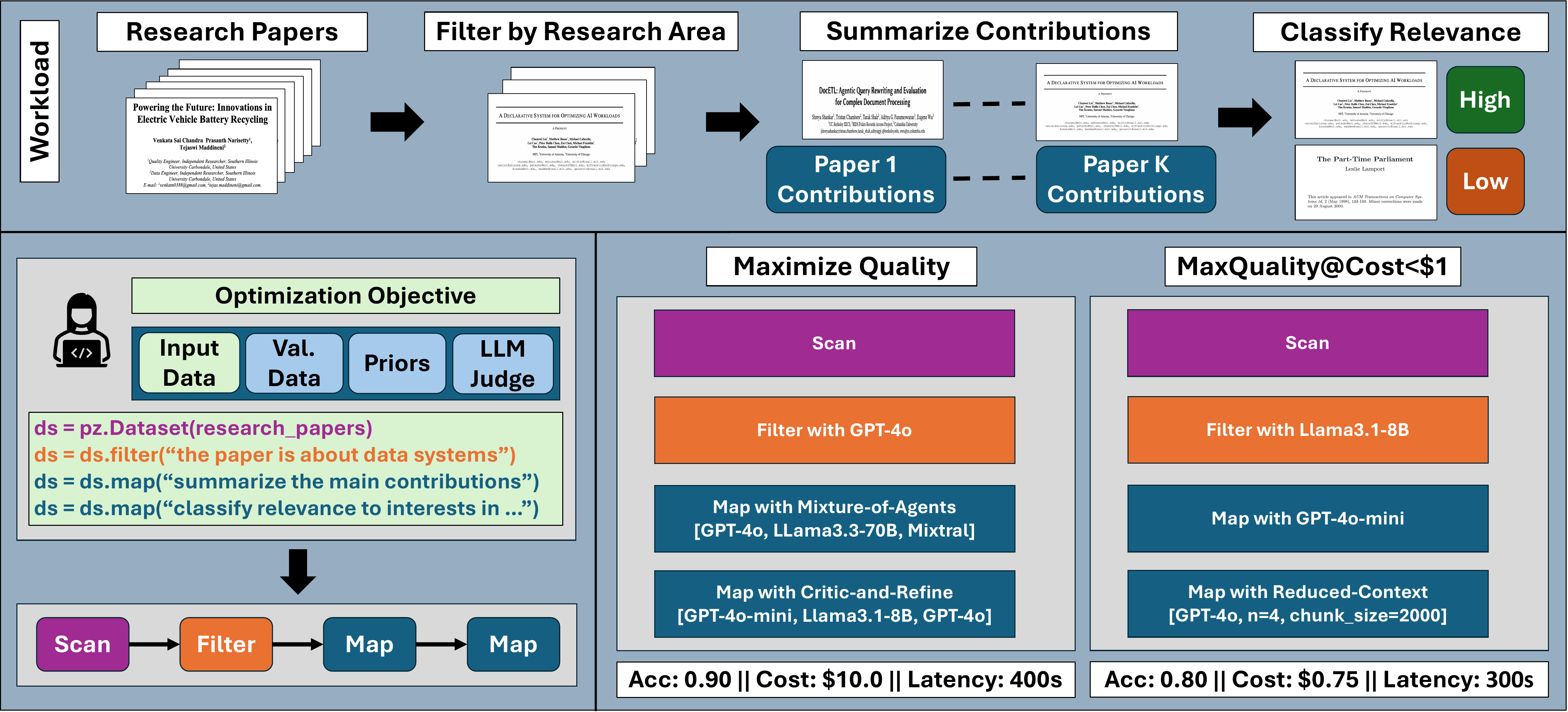}
\caption{An illustration of \system{} compiling a program for a literature search workload (Top) into two different physical plans for two different optimization objectives. (Left) the user implements the workload in a Palimpzest program and provides the input data they wish to process and (optionally) validation data, priors, and/or an LLM judge which \system{} may use to guide its optimization. The program is compiled to a logical plan, which \system{} seeks to implement with a (near-)optimal physical plan. (Center) given the unconstrained objective of maximizing quality, \system{} produces a physical plan which achieves high quality for this task. (Right) given the objective of maximizing quality subject to a constraint of \$1 in execution cost, \system{} produces a plan which satisfies the constraint while only trading-off a modest decrease in quality.}

\Description{An illustration of Abacus compiling a program for a literature search workload (Top) into two different physical plans for two different optimization objectives. (Left) the user implements the workload in a Palimpzest program and provides the input data they wish to process and (optionally) validation data which Abacus may use to guide its optimization. The program is compiled to a logical plan, which Abacus seeks to implement with a (near-)optimal physical plan. (Center) given the unconstrained objective of maximizing quality, Abacus produces a physical plan which achieves high quality for this task. (Right) given the objective of maximizing quality subject to a constraint of \$1 in execution cost, Abacus produces a plan which satisfies the constraint and is much cheaper than the unconstrained plan, while only trading-off a modest decrease in quality.}
\label{fig:example}
\end{figure*}

Industry and academia are increasingly using large language models (LLMs) to solve problems which require semantic understanding. These problems range from unstructured document processing \cite{doosterlinck2023biodex, hendrycks2021cuad}, to multi-modal question answering \cite{talmor2021multimodalqa, yue2023mmmu, song2024milebench}, to semantic search and ranking \cite{thakur2021beir}. In order to achieve state-of-the-art performance on these tasks, practitioners often decompose the problem into modular subtasks within an AI program.

Recently, programming frameworks including Palimpzest \cite{liu2025palimpzest}, LOTUS \cite{patel2025semanticoperatorsdeclarativemodel}, DocETL \cite{shankar2024docetlagenticqueryrewriting}, and others \cite{khattab2023dspycompilingdeclarativelanguage, lu2025vectraflow, saadfalcon2024archonarchitecturesearchframework, aryn2025, galois2025satriani} have proposed building these LLM-based applications out of \textbf{semantic operators}. Inspired by relational operators \cite{hellerstein2007}, semantic operators are AI-powered data transformations with natural language specifications. These include LLM-powered maps, filters, joins, aggregations, etc. and are useful for unstructured data processing tasks such as information extraction, summarization, ranking, and classification.

Developers can define a \textbf{semantic operator system} by writing a declarative AI program (e.g., in Palimpzest or a similar framework). The program defines a logical plan, which an optimizer can compile into a physical plan. For example, \Cref{fig:example} illustrates a use case where a researcher wishes to search for papers relevant to their interests. First, the program loads the papers and filters for ones related to data systems. Then, the program computes a summary of each paper's main contributions. Finally, the papers are classified as having high or low relevance to the author's research interests.

In order to execute this program, the optimizer must decide how to implement each semantic filter and map in terms of underlying physical operations (e.g., calls to LLMs). For example, given many LLMs of different sizes, the optimizer may simply need to choose which LLM to use for each operator. However, the optimizer may also be allowed to choose from more complex techniques such as using an LLM ensemble \cite{wang2025mixtureofagents}, reducing the input context before feeding it to an LLM \cite{lewis2020rag}, and more. With access to a handful of models and hyperparameters, a few techniques can provide an optimizer with thousands of physical implementation alternatives that trade-off operator quality, dollar cost, and latency (\Cref{sec:implementation}).

\textbf{Goal.} The optimizer's goal is to compile a semantic operator program to a physical plan which is (near-)optimal for the developer's objective with respect to system quality, cost, and latency. For example, in the center and right-hand side of \Cref{fig:example}, we show two physical plans for two different optimization objectives. The first plan is compiled with the goal of maximizing system output quality, while the second is compiled to maximize quality subject to spending less than \$1 on processing the entire workload. Here, the cost-constrained plan uses lighter weight models and simpler strategies.  Ideally, a cost-based optimizer can weigh the trade-offs of different physical operators to implement each plan optimally.

\textbf{Our Approach.} In this paper we describe \system{}, a new cost-based optimizer implemented in Palimpzest~\cite{liu2025palimpzest}. While some existing semantic programming frameworks have implemented optimizers, they typically only optimize for system quality, do not consider constraints on system cost or latency, and use a limited set of rewrite-style or proxy-based optimizations \cite{shankar2024docetlagenticqueryrewriting, patel2025semanticoperatorsdeclarativemodel}. By contrast, \system{} is a general-purpose, cost-based optimizer that optimizes system output quality, dollar cost, or latency with respect to zero or more constraints on the other dimensions. \system{} uses implementation and transformation rules to define a valid set of physical plans, similar to a Cascades query optimizer \cite{Graefe1995TheCF}. However, the uncertain nature of semantic operator quality---combined with the lack of principled models (e.g. cardinality estimators, histograms, etc.) for estimating operator quality, cost, and latency---makes cost-based optimization challenging. 

To overcome this, \system{} applies three key ideas. First, \system{} models the search for useful operators as an infinite-armed bandit problem \cite{agrawal1995infinite, audibert2008mabsurvey} and uses sampling to estimate physical operator performance. Given the cost of invoking LLMs, \system{} must be judicious in choosing which physical operators to sample and how many samples to spend on each operator. This is especially difficult in constrained optimization settings, where \system{} must discover the Pareto frontier of physical operators as opposed to a single objective maximizing operator. To this end, \system{} modifies an upper-confidence bound (UCB) bandit algorithm to enable it to search for the Pareto frontier of physical operators. The multi-armed bandit (MAB) algorithm can also leverage prior beliefs about operator performance to significantly accelerate its search.

Second, similar to relational query optimization, the space of physical plans grows combinatorially with the number of operators in the system. However, while relational query optimizers can use precomputed statistics to estimate the performance of plans at scale, \system{}'s sample-based approach quickly becomes too expensive. To mitigate this issue, \system{} approximates plan performance as a function of its individual operators' performance. This decomposition allows \system{} to estimate the performance of a combinatorially large space of plans given a much smaller set of operator estimates. Third, the traditional dynamic programming algorithm used in Cascades \cite{Graefe1995TheCF} is not designed to support constrained optimization problems. To overcome this, \system{} implements a new Pareto-Cascades algorithm which keeps track of the Pareto frontier of subplans throughout the optimization procedure.  

\textbf{Performance.} We have implemented \system{} as an optimizer in Palimpzest and evaluate its ability to optimize systems for document processing workloads in the biomedical and legal domains (BioDEX; CUAD) and multi-modal question answering (MMQA). Our results show that \system{} identifies plans with 20.8\%, 39.4\%, and 6.7\% better quality, respectively, than similar plans optimized by DocETL and LOTUS. Furthermore, plans optimized by \system{} are on-average 10.8x cheaper and 3.4x faster than plans optimized by the next best system (in terms of quality).

We also show that, at a fixed sample budget, \system{} can use prior beliefs (i.e., a relative ranking of operators' quality, cost, and latency) to optimize plans to have up to 3.04x better quality than without priors. We further demonstrate that \system{} satisfies constraints in a non-trivial manner and improves system performance as constraints are relaxed. Finally, we perform an ablation study to isolate the benefits of prior beliefs, Pareto-Cascades, and the MAB algorithms and show that each helps improve \system{}'s performance on two constrained optimization queries.

In summary, we present \system{} --- a cost-based optimizer for semantic operator systems. Our main contributions are:
\begin{itemize}
    \item An extensible, cost-based optimizer which allows for new semantic operators and optimization rules without changes to its host programming framework (\Cref{sec:system-overview}).
    \item The implementation of algorithms which enable (1) efficient search over the space of semantic operator systems and (2) constrained optimization of these systems (\Cref{sec:algorithms}).
    \item Quality improvements of up to 39.4\% over competing state-of-the-art systems with cost and runtime savings of 10.8x and 3.4x relative to the next best system (\Cref{sec:evaluation}).
    \item An investigation of \system{}'s algorithmic contributions which shows that prior beliefs, Pareto-Cascades, and MAB sampling can improve optimization outcomes (\Cref{sec:evaluation}).
\end{itemize}

%% file: sections/sys_overview.tex
\section{System Overview}
\label{sec:system-overview}
In this section we present an overview of \system{}. First, we provide a brief background on semantic operators and the programming frameworks which optimize them. Then, we describe the end-to-end process by which \system{} optimizes semantic operator systems. Finally, we motivate the need for two key algorithms to make \system{}'s optimization tractable, which we discuss in \Cref{sec:algorithms}.

\subsection{Background: Semantic Operator Systems}
\textbf{Background and terminology.} Recent work has explored the use of \textit{semantic operators} to implement data processing pipelines over unstructured data. Semantic operators are a set of AI-powered data transformations which mirror and extend relational operators \cite{hellerstein2007}. Unlike their relational counterparts, semantic operators are specified in natural language as opposed to a SQL expression or a user-defined function. As a result, these operators' physical implementations typically require the use of one or more foundation models with semantic understanding.

\begin{table}[t]
  \caption{Semantic operators supported by \system{}. In our implementation $d$ is a (valid) JSON dictionary, but in principle $d$ can be any serializable object. The $\cup$ symbol represents the union of output types. $i$ is an integer index, $P$ is a filter predicate, $V$ is a vector database, and $L$ is an integer limit. Aggregate includes group-by operations.}
  \label{tab:operators}
  \begin{tabular}{ccl}
    \toprule
    Operator Name & Symbol & Definition  \\
    \midrule
    Scan & $\phi$ & $\phi(i) \rightarrow d$ \\
    Map & $\mu$ & $\mu(d) \rightarrow d' \cup [d', d'', \dots]$ \\
    Filter & $\sigma$ & $\sigma(d, P) \rightarrow d \cup \emptyset$ \\
    Join & $\bowtie$ & $\bowtie(d,d') \rightarrow d'' \cup \emptyset$ \\
    Top-K & $\rho$ & $\rho(d, V) \rightarrow d' $ \\
    Project & $\pi$ & $\pi(d) \rightarrow d' \subseteq d$ \\
    Aggregate & $\alpha$ & $\alpha([d', d'', \dots]) \rightarrow \mathbb{R} \cup d$\\
    Limit & $\lambda$ & $\lambda([d', d'', \dots], L) \rightarrow [d', \dots, d^L]$\\
  \bottomrule
\end{tabular}
\end{table}

Semantic programming frameworks like Palimpzest \cite{liu2025palimpzest}, LOTUS \cite{patel2025semanticoperatorsdeclarativemodel}, DocETL \cite{shankar2024docetlagenticqueryrewriting}, and Aryn \cite{aryn2025} enable users to compose semantic operators into pipelines or directed acyclic graphs (DAGs). We refer to these computation graphs of semantic operators as \textit{semantic operator systems}. Each framework implements an evolving and growing set of semantic operators, thus we highlight the operators currently supported by \system{} in \autoref{tab:operators}.

Each semantic operator corresponds to a \textbf{logical operator} which may be implemented by a variety of \textbf{physical operators}. For example, two of the semantic map operators in \autoref{fig:example} are implemented with a Mixture-of-Agents \cite{wang2025mixtureofagents} architecture and a Reduced-Context generation (\Cref{sec:implementation}). The former is a layered computation graph of LLM ensembles, while the latter reduces the context to contain only the most relevant portions of the input before feeding it into an LLM. Each of these physical operators has multiple hyperparameters (e.g. the models and temperature settings for Mixture-of-Agents; the model, chunk size, and number of chunks for Reduced-Context) leading to a large space of physical operators. 

\subsection{\system{} Optimizer}
\label{sec:abacus-system-description}
We illustrate \system{}'s end-to-end process for optimizing semantic operator systems, beginning with a high-level overview of its key steps which are shown in \Cref{fig:system-diagram}.

\begin{figure*}[t!]
    \centering
    \includegraphics[width=\linewidth]{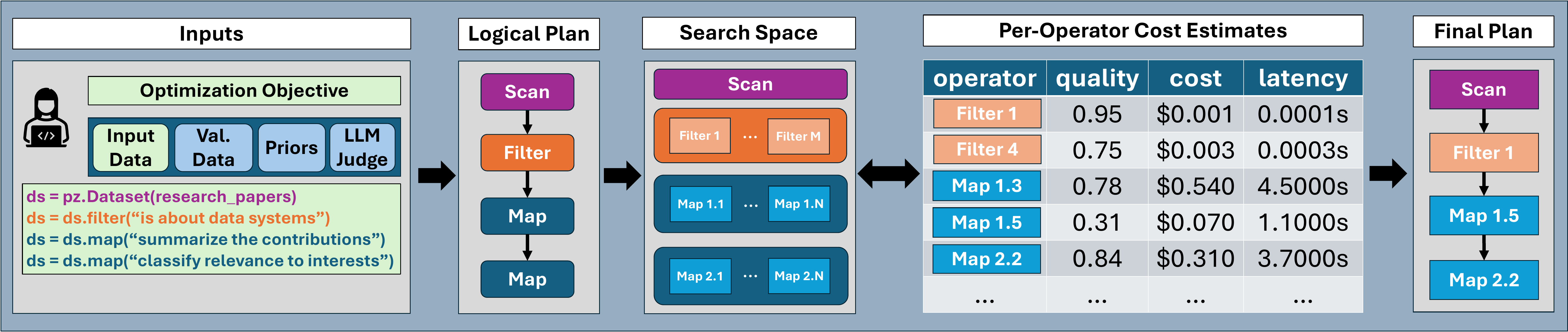}
    \caption{Overview of \system{}. The developer provides an AI program, optimization objective, input data, and (optionally) validation data, priors, and/or an LLM judge. \system{} (1) compiles the program to an initial logical plan, (2) applies rules to enumerate a search space of physical plans, (3) builds a cost model by processing validation inputs with sampled physical operators to measure their performance, and (4) returns the Pareto-optimal plan based on its estimates and the user objective.}
    \Description{Overview of Abacus. The developer provides an AI program, optimization objective, input data, and (optionally) validation data. The program is compiled into an initial logical plan. Abacus applies rules to enumerate a search space of physical plans. Abacus iteratively samples physical operators and processes validation inputs with them to build up a cost model of the operator performance. Abacus returns the Pareto-optimal plan based on its estimates and the user objective.}
    \label{fig:system-diagram}
\end{figure*}

\textbf{Inputs and Compilation.} \system{} requires three inputs: an AI program, an optimization objective, and an input dataset. The AI program must be a pipeline or DAG of semantic operators supported by \system{}. The optimization objective is a constrained or unconstrained objective with respect to system output quality, dollar cost, and/or latency. The input dataset is an unstructured dataset of documents, images, songs, etc. which the physical implementation of the AI program will process. Optionally, users may guide the optimization process by providing any combination of a small validation dataset, prior beliefs about operator performance, and/or an LLM to use as a judge. A typical validation dataset contains 5-10 inputs with (possibly partial) labels. Prior beliefs are a simple dictionary mapping operators to a three-tuple of their perceived quality, cost, and latency on a [0,1] scale. Finally, \system{} can use an LLM as a judge to evaluate the quality of physical operators' outputs when labels are not present.

For example, in \Cref{fig:example}, the AI program consists of a semantic filter followed by two semantic maps. The figure shows two objectives: maximizing quality and maximizing quality subject to a constraint on cost. The input dataset is a set of research papers, and the validation dataset (not shown) could be a handful of additional research papers whose relevance has been labeled. Finally, given these inputs, \system{} compiles the program into a logical plan where each semantic operator corresponds to a logical operator. 

\textbf{Creation of Search Space.} Once the user's program has been compiled to a logical plan, \system{} uses its rules to enumerate a space of valid physical operators for each logical operator. This corresponds to the Search Space in \Cref{fig:system-diagram}. Each rule consists of two parts: (1) a pattern matching function which defines the logical subplan the rule can be applied to, and (2) a substitution function which applies the rule.

\textit{Transformation rules} produce new functionally equivalent logical subplans. For example, a transformation rule may swap a filter and a map operation such that the filter is executed before the map. \textit{Implementation rules} define ways to implement semantic operator(s) in a logical plan. For example, an implementation rule may implement a map operator with a Mixture-of-Agents or a Reduced-Context generation as depicted in \Cref{fig:example}. Currently, \system{} only applies implementation rules when generating the search space, but in the future it may apply transformation rules as well. \system{} will apply transformation rules during the Final Plan Selection (see below) to ensure the final physical plan makes use of optimizations like filter pushdown and join re-ordering.

\textbf{Operator Sampling.} Given the search space of physical operators, \system{} seeks to identify ones which can be composed into physical plans that optimize the user's objective. For unconstrained optimization (e.g. maximizing plan quality), this implies finding high-quality physical operator(s). For optimization with constraints (e.g. maximizing plan quality subject to a cost constraint), this suggests finding physical operators which lie on the Pareto frontier of the cost vs. quality trade-off.

\system{} initially samples a small batch of physical operators for each logical operator. If \system{} has access to prior beliefs about operator performance, it samples operators which it believes lie closest to the Pareto frontier of the optimization objective first. Otherwise, it samples operators at random. Given these operators, \system{} executes them on inputs sampled from the validation dataset (or the input dataset if no validation data is present). \system{} measures the quality, cost, and latency of each operator on each input. To measure quality, \system{} uses output label(s) from the validation dataset when they are available. However, if no label exists---either because validation data is not provided or it does not contain some intermediate label(s)---then \system{} evaluates each operator's output with an LLM judge. By default, \system{} will use \texttt{o4-mini} as the LLM judge, but the user may specify another model.

Once it has estimated each operator's quality, cost, and latency, \system{} computes the Pareto frontier of physical operators (with respect to the optimization objective) for each logical operator. Physical operators which fall too far from the frontier are removed, and new operators are sampled to replace them. The next batch of inputs is then processed with the new operator frontiers, and the process repeats until the sample budget (measured in dollars or operator invocations) has been reached. The cost overhead of the sampling algorithm is bounded above by the sample budget, and the latency overhead only scales with the depth of the logical plan.

\textbf{Final Plan Selection.} Once the sample budget is exhausted, \system{} needs to construct a final plan to process the input dataset. First, it computes each physical operator's average quality, cost, and latency on sampled inputs. \system{} then passes these estimates and the user's optimization objective to its Pareto-Cascades algorithm (\Cref{sec:pareto-cascades}) to compute the optimal plan. The algorithm applies both transformation and implementation rules to search the full space of logical and physical plans, using \system{}'s cost model (\autoref{sec:optimization-challenges}) to assess their quality, cost, and latency.

\textbf{Full Algorithm.} The full algorithm for \system{} is shown in \Cref{algo:abacus}. The user program is compiled into an initial logical plan on line 1. On line 2, \system{} applies rules to create a search space of physical operators. Line 3 initializes a cost model which keeps track of each operator's average quality, cost, and latency. We describe the cost model in more detail in \Cref{sec:optimization-challenges}. On line 4, \system{} samples an initial ``frontier" of $k$ physical operators for each logical operator. On line 7, each frontier processes a sample of $j$ inputs, updating the number of samples drawn. This also yields a set of observations of operator quality, cost, and latency, which are used to update the cost model on line 8. On line 9, operators which perform poorly are replaced in each frontier. We discuss the algorithm for updating the operator frontiers in detail in \Cref{sec:mab-sampling-algo}. Once the number of samples drawn (or the dollar cost of sampling) exceeds the sample budget on line 6, the operator sampling stops. Finally, on line 10 \system{}'s Pareto-Cascades algorithm returns the optimal physical plan with respect to the the operator estimates and the optimization objective. We discuss the Pareto-Cascades algorithm in detail in \Cref{sec:pareto-cascades}.

\subsection{Key Challenges in Optimization}
\label{sec:optimization-challenges}
We now motivate the design of \system{}'s cost model, operator sampling algorithm, and final plan selection algorithm.

\textbf{Cost Model.} Given a logical plan with $M$ semantic operators and a choice of $N$ physical implementations per operator, the space of possible physical plans is of size $O(N^M)$ before considering operator re-orderings. Even for relatively modest values of $M$ and $N$, the space of plans quickly grows too large to sample each plan and measure its output quality, cost, and latency.

To address this, \system{} makes the simplifying assumption that operators are independent, and that each plan can be modeled as a function of its operators. \system{}'s model for the plan quality ($p_q$), cost ($p_c$), and latency ($p_l$) as a function of its operators' quality ($o_{qi}$), cost ($o_{ci}$), and latency ($o_{li}$) is shown below:
\begin{align}
    \hat{p}_{q} &= \prod_{i=1}^M \hat{o}_{qi} & \hat{p}_{c} &= \sum_{i=1}^M \hat{o}_{ci} & \hat{p}_{l} &= \max_{\textrm{path} \in p} \sum_{i \in \textrm{path}} \hat{o}_{li} \label{eq:cost-model}
\end{align}

One limitation of this cost model is that it fails to model interactions between operators. For example, if a semantic filter uses the summary produced by an upstream map operator as input, then this cost model will fail to capture that the filter's performance is correlated with the quality of the map operator's summary. We evaluate \system{} on some queries with this property in \Cref{sec:abacus-vs-prior-work} and discuss the limitations of this cost model in \Cref{sec:limitations}.

\begin{algorithm}[t]
    \caption{\system{} algorithm}
    \label{algo:abacus}
    \begin{algorithmic}[1]
        \Require program $P$, objective $O$, val. data $D$
        \Params budget $B$; $k$, $j$
        \State $logical\_plan$ = compile($P$)
        \State $search\_space$ = applyRules($logical\_plan$)
        \State $M$ = initCostModel()
        \State $F$ = sampleOpFrontiers($search\_space$, $k$)
        \State $samples\_drawn = 0$
        \While{$samples\_drawn < B$}
            \State $outputs, samples\_drawn$ = processSamples($F$, $D$, $j$) \State $M$ = updateCostModel($M$, $outputs$)
            \State $F$ = updateFrontiers($F$, $M$, $O$)
        \EndWhile
        \State \Return ParetoCascades($logical\_plan$, $M$, $O$)
    \end{algorithmic}
\end{algorithm}

\textbf{Operator Sampling Challenges.} For large enough $N$, it can be computationally infeasible to sample every physical operator for even a single semantic operator. For example, in our implementation of \system{} (\Cref{sec:implementation}), a semantic map can be implemented with approximately 2,800 different operators. While the task of finding and choosing physical operator(s) may seem daunting, in most settings \system{} simply needs to produce a plan which is ``good enough" for the user's application goals. This relaxes the operator search problem from finding the single best ``needle in a haystack" to finding at least one operator from a handful of good options.

In the $sampleOpFrontiers()$ function in \Cref{algo:abacus}, we sample an initial set (i.e., frontier) of physical operators for each logical operator in the plan. Then, in the $updateFrontiers()$ function we update each frontier of physical operators based on their observed quality, cost, and latency on sampled inputs. We model the sampling of physical operators as a multi-armed bandit (MAB) problem. Intuitively, given a fixed sampling budget, we seek to navigate the exploration-exploitation trade-off between sampling new (potentially better) operators and sampling the previously best observed operator(s) to refine our confidence in their performance.

Unfortunately, the traditional MAB formulation is focused on finding the single most-optimal arm for an unconstrained objective. However, constrained optimization requires that we account for the trade-off between the optimization objective and the constraint(s). To this end, we modify the traditional MAB formulation to encourage the exploration-exploitation trade-off of \textit{the entire Pareto frontier of operators}. We formalize this algorithm in \Cref{sec:mab-sampling-algo}.

\textbf{Final Plan Selection Challenges.} Once \system{} finishes sampling operators, it still needs to identify and return the optimal physical plan. For an unconstrained objective such as minimizing plan cost, \system{} can invoke a traditional Cascades \cite{Graefe1995TheCF} algorithm to recover the minimum cost plan.

However, for constrained optimization the traditional Cascades algorithm is insufficient. The key issue is that Cascades will only keep track of the ``best" implementation of every subplan. However, in the constrained setting---where we care about multiple dimensions of plan performance---finding the optimal plan requires considering the Pareto frontier of optimization trade-offs at each subplan. We implement the Pareto-Cascades algorithm to overcome this challenge, and discuss its implementation in \Cref{sec:pareto-cascades}.

%% file: sections/algorithms.tex
\begin{figure*}[!t]
    \centering
    \includegraphics[width=.9\textwidth]{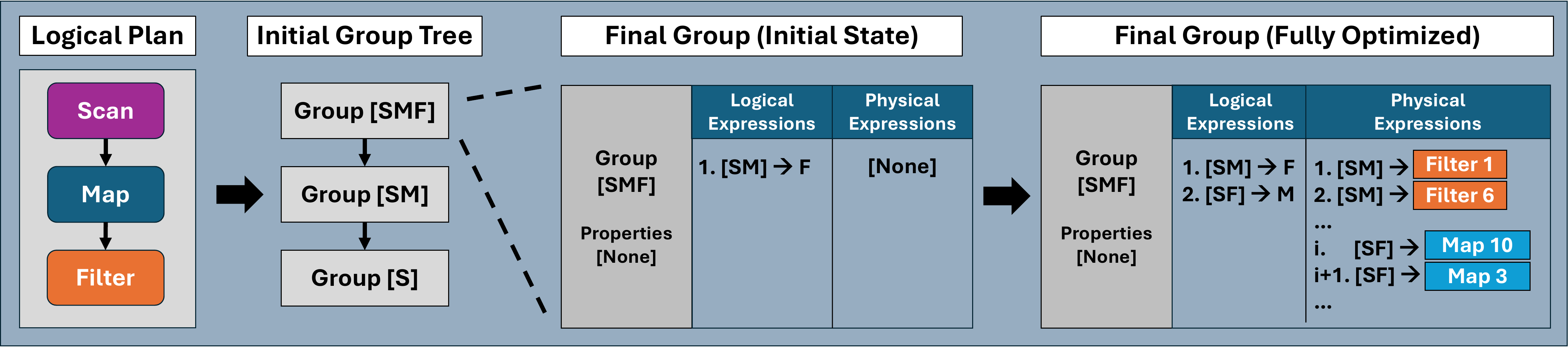}
    \caption{Example of the Cascades algorithm applied to a simple logical plan; it first constructs an initial group tree with one logical expression per group and then applies a task to optimize the final group (\textbf{SMF}), which uses dynamic programming to search plans via repeated application of tasks. After all possible tasks have been applied (or a limit on the total number of tasks has been reached), it recursively constructs the optimal physical plan by selecting the best physical expression at each group.}
    \Description{An example of the Cascades algorithm applied to a simple logical plan.}
    \label{fig:cascades-example}
\end{figure*}
\begin{algorithm}[t]
    \caption{Full Cascades algorithm}
    \label{algo:traditional-cascades}
    \begin{algorithmic}[1]
        \Require logical plan $P$, cost model $M$, rules $R$
        \Procedure{Cascades}{$P$, $M$, $R$}
            \State $G$ = createInitialGroups($P$)
            \State $G$ = searchPlanSpace($G$, $M$, $R$)
            \State $final\_group\_id$ = getFinalGroupId($G$)
            \State \Return getMinCostPlan($final\_group\_id$, $G$)
        \EndProcedure
        \Procedure{getMinCostPlan}{$group\_id$, $G$}
            \State $best\_expr = G[group\_id].best\_expr$
            \State $input\_group\_id = best\_expr.input\_group\_id$
            \If{$input\_group\_id$ is None}
                \State \Return Plan($best\_expr.operator$)
            \EndIf
            \State $best\_subplan$ = getMinCostPlan($input\_group\_id$, $G$)
            \State \Return Plan($best\_subplan$, $best\_expr.operator$)
        \EndProcedure
    \end{algorithmic}
\end{algorithm}
\begin{algorithm}[t]
    \caption{Cascades Subroutine: Plan Search}
    \label{algo:cascades-plan-search}
    \begin{algorithmic}[1]
        \Require groups $G$, cost model $M$, rules $R$
        \Procedure{searchPlanSpace}{$G$, $M$, $R$}
        \State $task\_stack = [\textrm{OptimizeGroup(getFinalGroupId($G$))}]$
        \While{len($task\_stack$) > 0}
            \State $task = task\_stack.pop()$
            \State $new\_tasks = task.perform(G, M, R)$
            \For{$new\_task$ in $new\_tasks$}
                \State $task\_stack.push(new\_task)$
            \EndFor
        \EndWhile
        \State \Return $G$
        \EndProcedure
    \end{algorithmic}
\end{algorithm}

\section{Algorithms}
\label{sec:algorithms}
In this section, we first present a high-level overview of the traditional Cascades algorithm from relational query optimization. We then discuss the Pareto-Cascades and multi-armed bandit (MAB) operator sampling algorithms we developed for \system{}.

\subsection{Traditional Cascades Optimization}
\label{sec:traditional-cascades}
Cascades \cite{Graefe1995TheCF, xu1998columbia} takes a logical plan, a cost model, and a set of rules as input. Given these inputs, Cascades seeks to find an implementation of each operator that globally optimizes the plan to meet some objective (e.g. minimizing execution cost).

Consider the toy example in \Cref{fig:cascades-example}. First, Cascades converts the logical plan into an \textbf{initial group tree} (\Cref{algo:traditional-cascades}, Line 2). Each \textbf{group} represents the execution of a unique set of operators. In \Cref{fig:cascades-example}, we expand the \textbf{final group} which represents the execution of all operators in the plan. Each group has a set of \textbf{logical} and \textbf{physical expressions}, which represent unique logical and physical subplans which implement that group.

Initially, each group has a single logical expression which is translated directly from the logical plan (in this case, executing filter $F$ after map $M$ and scan $S$). Given the initial group tree, the Cascades algorithm searches the space of possible physical plans (\Cref{algo:traditional-cascades}, Line 3) by applying a series of \textbf{tasks} in a dynamic programming algorithm shown in \Cref{algo:cascades-plan-search}. There are four main tasks: Optimize Group, Optimize Logical Expression, Apply Rule, and Optimize Physical Expression. The call to optimize a group triggers tasks for optimizing each of its logical and physical expressions. Logical expressions are optimized by applying transformation rules (to generate new equivalent logical expressions) and implementation rules (to generate physical expressions). Finally, the task to optimize a physical expression computes the minimum cost plan for executing the given expression.

We present the full Cascades algorithm in \Cref{algo:traditional-cascades}. The algorithm takes a logical plan, a cost model, and a set of rules as input. It constructs the initial groups (i.e, the group tree) on Line 2 and then invokes the plan search procedure in \Cref{algo:cascades-plan-search} on Line 3. Once the search finishes, the group tree is traversed to construct the final physical plan by using the physical operator in the optimal (i.e. min. cost) physical expression for each group (Line 5). With this understanding in place, we will now discuss \system{}'s Pareto-Cascades algorithm.

\subsection{Pareto-Cascades Optimization}
\label{sec:pareto-cascades}
As discussed in \Cref{sec:optimization-challenges}, Cascades is not designed to handle optimization problems with constraints. The key issue is the \textbf{Principle of Optimality}, which states that \textit{every subplan of an optimal physical plan is itself optimal}. This principle enables Cascades to optimize each physical expression by composing it with the optimal expression for its input group. This is insufficient for problems such as minimizing cost with a lower bound on plan quality, because selecting the minimum cost expression for each group may result in constructing a plan that fails to meet the quality constraint.

In order to address this issue, each dimension of the optimization problem (e.g., cost and quality) must be accounted for. Fortunately, there is a natural way to extend the Principle of Optimality into the constrained optimization setting, which we present as a theorem:

\begin{theorem}
(Under the operator independence assumptions of our cost model in \Cref{sec:optimization-challenges}) every subplan of a Pareto-optimal physical plan is itself Pareto-optimal.
\end{theorem}
    
\begin{proof}
    We prove this by contradiction. Assume a Pareto-optimal physical plan $P$ has a subplan $S$ which is not Pareto-optimal. By the definition of $S$ not being Pareto-optimal, there exists a subplan $S'$ which dominates $S$. Replacing $S$ with $S'$ strictly improves the quality, cost, and/or latency of the subplan. Given the operator independence assumptions of our cost model in \Cref{eq:cost-model}, strictly improving the subplan will also strictly improve the quality, cost, and/or latency of the entire physical plan. This new physical plan $P'$ will be strictly better than our original plan $P$ -- but this contradicts our assumption that the original plan $P$ is Pareto-optimal.
\end{proof}

This theorem enables us to extend the Cascades algorithm to the constrained optimization setting by modifying each group to maintain its Pareto frontier of physical expressions during the search procedure in \Cref{algo:cascades-plan-search}. For example, if a user's objective is to maximize plan quality with an upper bound on plan cost, then each group maintains its set of physical expressions which are Pareto-optimal with respect to quality and cost. The task to optimize physical expressions is modified to compute the Pareto frontier of executing the current physical expression with any of the Pareto-optimal expressions from its input group(s). Finally, once the search procedure is finished, the Pareto-optimal plan is recovered by recursively composing all Pareto-optimal subplans before selecting the final plan which is optimal for the given optimization objective. (While the Pareto frontier introduces a branching factor to the plan search space, it is bounded above by the number of physical operators sampled for a given semantic operator).

We present the algorithm for our new Pareto-Cascades algorithm in \Cref{algo:pareto-cascades}. We use the same function for searching the plan space with the modifications described in the previous paragraph. The $getParetoOptPlans()$ function is similar in spirit to $getMinCostPlan()$ in \Cref{algo:traditional-cascades}, except it builds and returns a list of Pareto-optimal plans. The $selectOptimalPlan()$ function picks the plan on the Pareto frontier which is optimal for the optimization objective $O$ (e.g. selecting the max quality plan which is cheaper than a cost upper bound). If the Pareto-Cascades algorithm cannot find a plan which satisfies the given constraint, then the algorithm will return the plan which best optimizes the given objective. Finally, we note that in the case of unconstrained optimization, this algorithm naturally reduces to the traditional Cascades algorithm.

\begin{algorithm}[t]
    \caption{Pareto-Cascades algorithm}
    \label{algo:pareto-cascades}
    \begin{algorithmic}[1]
        \Require logical plan $P$, cost model $M$, rules $R$, objective $O$
        \Procedure{ParetoCascades}{$P$, $M$, $R$, $O$}
            \State $G$ = createInitialGroups($P$)
            \State $G$ = searchPlanSpace($G$, $M$, $R$)
            \State $final\_group\_id$ = getFinalGroupId($G$)
            \State $pareto\_plans$ = getParetoOptPlans($final\_group\_id$, $G$)
            \State \Return selectOptimalPlan($pareto\_plans$, $O$)
        \EndProcedure
    \end{algorithmic}
\end{algorithm}

\subsection{Multi-Armed Bandit Operator Sampling}
\label{sec:mab-sampling-algo}
The second key optimization challenge in \system{} is choosing which physical operators to sample in order to obtain estimates of operator quality, cost, and latency. As discussed in \Cref{sec:optimization-challenges}, we assume that the number of physical operators $N$ is large enough that \system{} cannot realistically sample every physical operator. To overcome this issue, we draw inspiration from the infinite-armed bandit problem \cite{agrawal1995infinite, audibert2008mabsurvey}, which can also serve as a model for settings with more arms than total samples.

In our setting, the physical operators comprise the ``arms" of our search space and we are given an initial sample budget $B$. At each step of the search, we must choose a physical operator to sample (decreasing our budget by one or by the cost of invoking the operator if $B$ is in dollars) and obtain a stochastic observation of that operator's performance. In contrast to the traditional multi-armed bandit (MAB) setting, where the objective is to identify the single best arm achieving the highest performance in expectation, \system{}'s goal is to identify the potentially many physical operators which lie on the Pareto frontier of its optimization objective.

We present \system{}'s MAB operator sampling algorithm in \Cref{algo:mab-sampling}. The inputs to the algorithm are an initial set of physical operator frontiers $F$ (one for each logical operator, from line 4 of \Cref{algo:abacus}), the cost model $M$, and the optimization objective $O$. The algorithm begins by computing the upper confidence bounds (UCBs), lower confidence bounds (LCBs), and means for each operator on each metric of interest for the objective $O$. The equations for computing the UCB and LCB of a given metric are shown below: 
\begin{align*}
    ucb_{m,i} &= \mu_{m,i} + \alpha \cdot \sqrt{\frac{\log(N)}{n_i}} & lcb_{m,i} &= \mu_{m,i} - \alpha \cdot \sqrt{\frac{\log(N)}{n_i}}
\end{align*}
\noindent The $\mu_{m,i}$ term is the sample mean of the observed performance for the given metric $m$ (e.g. operator latency) for the $i^{th}$ physical operator. $N$ is the total number of samples drawn and $n_i$ is the number of samples drawn for the $i^{th}$ physical operator. Finally, $\alpha \in [0, 1]$ is the exploration coefficient, which we dynamically scale to be $0.5$ times the spread between the largest and smallest observed metric values across all physical operators.

Once the UCBs, LCBs, and means are computed for every operator and metric, we compute the set of Pareto-optimal operators based on their mean performance. Then, for each operator in the frontier, we check whether its upper confidence bound overlaps with the lower confidence bound of at least one operator on the Pareto frontier. Such an overlap implies that there's enough uncertainty in our estimates of operator performance that it is possible for the operator to lie on the Pareto frontier. If no overlap exists, then we remove the operator from the frontier and sample a replacement from our reservoir of not yet sampled physical operators. This completes the update of the operator frontier.

The key difference between this algorithm and a traditional UCB algorithm for MABs is that we must consider overlap between each operator and the current Pareto frontier of sampled operators. Overlap on any dimension implies that the operator may still be Pareto-optimal, thus eliminating operators from consideration can be slightly more sample intensive and time consuming. In order to speed up the algorithm, we construct batches of samples rather than processing one sample at a time.

\begin{algorithm}[t]
    \caption{MAB operator sampling algorithm}
    \label{algo:mab-sampling}
    \begin{algorithmic}[1]
        \Require initial operator frontiers $F$, cost model $M$, objective $O$
        \Procedure{updateFrontiers}{$F$, $M$, $O$}
            \For{$op\_frontier$ in $F$}
                \State $op\_UCBs$ = computeUCBs($op\_frontier$, $M$, $O$)
                \State $op\_LCBs$ = computeLCBs($op\_frontier$, $M$, $O$)
                \State $op\_means$ = computeMeans($op\_frontier$, $M$, $O$)
                \State $pareto\_ops$ = computeParetoOps($op\_means$, $O$)
                \State $num\_new\_ops = 0$
                \For{$op$ in $op\_frontier$}
                    \State $ucbs = op\_UCBs[op]$
                    \State $pareto\_lcbs = op\_LCBs[pareto\_ops]$
                    \If{no\_overlap($ucbs$, $pareto\_lcbs$)}
                        \State $op\_frontier.pop(op)$
                        \State $num\_new\_ops \mathrel{+}= 1$
                    \EndIf
                \EndFor
                \State $new\_ops$ = sampleReservoir($num\_new\_ops$)
                \State $op\_frontier.add(new\_ops)$
            \EndFor
        \EndProcedure
    \end{algorithmic}
\end{algorithm}

Finally, one benefit of this problem formulation is that it allows for a number of extensions which can accelerate \system{}'s search for Pareto-optimal operators. For example, if there exist prior beliefs about operator performance, \system{} can use them to inform its initial operator frontier as well as the next operator(s) it draws from the reservoir during replacement. We explore the benefit of prior beliefs on operator performance in section \Cref{sec:priors}.

%% file: sections/evaluation.tex
\section{Evaluation}
\label{sec:evaluation}
We evaluate \system{} on a diverse set of benchmarks to examine three experimental claims. First, we demonstrate that semantic operator systems optimized by \system{} outperform similar systems optimized by prior work. Second, we show \system{} leverages prior beliefs to identify better plans with fewer samples.  Third, we demonstrate \system{} improves system performance as constraints are relaxed. Finally, we perform an ablation study to isolate the effects of our key algorithmic contributions.

\subsection{Implementation}
\label{sec:implementation}
We implement \system{} as an optimizer inside of the open-source Palimpzest \cite{liu2025palimpzest} framework, which supports all of the semantic operators in \Cref{tab:operators}. We wrote standard implementation rules for each semantic operator to provide \system{} with the ability to implement any Palimpzest program. We also implemented the following rules for optimizing map, filter, join, and top-k operators.

\textbf{Map and Filter.} We wrote four implementation rules which can be applied to map and filter operations. The \textbf{Model Selection} rule implements the operator with a single LLM call and is parameterized by the set of models supported by Palimpzest. The \textbf{Mixture-of-Agents} rule implements a Mixture-of-Agents architecture \cite{wang2025mixtureofagents} consisting of an ensemble of proposer models followed by an aggregator model. The rule is parameterized by (1) the size of the ensemble (1-3 proposers), (2) the model used for each proposer, (3) the model used for the aggregator, and (4) the temperature of the proposers (0.0, 0.4, or 0.8). The \textbf{Reduced-Context Generation} rule chunks the input, computes an embedding for each chunk, and then concatenates the top-k embeddings (based on similarity with the map/filter instruction) and feeds them into the operator. The rule is parameterized by the chunk size (1000, 2000, or 4000 characters) and $k$ (1, 2, or 4). Finally, the \textbf{Critique-and-Refine} rule uses an LLM to generate an initial output, which is then critiqued by a second model, before a third and final model generates a refined output. The rule is parameterized by the model used for each step.

\textbf{Top-K and Join.} We wrote a single rule to implement a top-k operator. The rule is parameterized by the value $k$ which determines the number of objects returned by the operator. We wrote two rules to implement a semantic join. The \textbf{Nested Loops Join} rule implements the join by evaluating the join condition with an LLM for every join tuple. The \textbf{Embedding Join} rule implements the join by automatically dropping tuples with low embedding similarity and automatically joining tuples with high embedding similarity (tuples in-between a high and low threshold are still processed with an LLM). The Nested Loops Join is generally expensive and accurate, while the Embedding Join rule is cheaper but less accurate.

These rules give \system{}  $\sim3,000$ physical operators when configured with access to all supported LLMs. For our experiments, unless stated otherwise, we provided \system{} with access to GPT-4o, GPT-4o-mini, Llama-3.1-8B, Llama-3.3-70B, Mixtral-8x7B, and DeepSeek-R1-Distill-Qwen-1.5B.

\begin{figure}[t!]
    \centering
    \includegraphics[width=.8\linewidth]{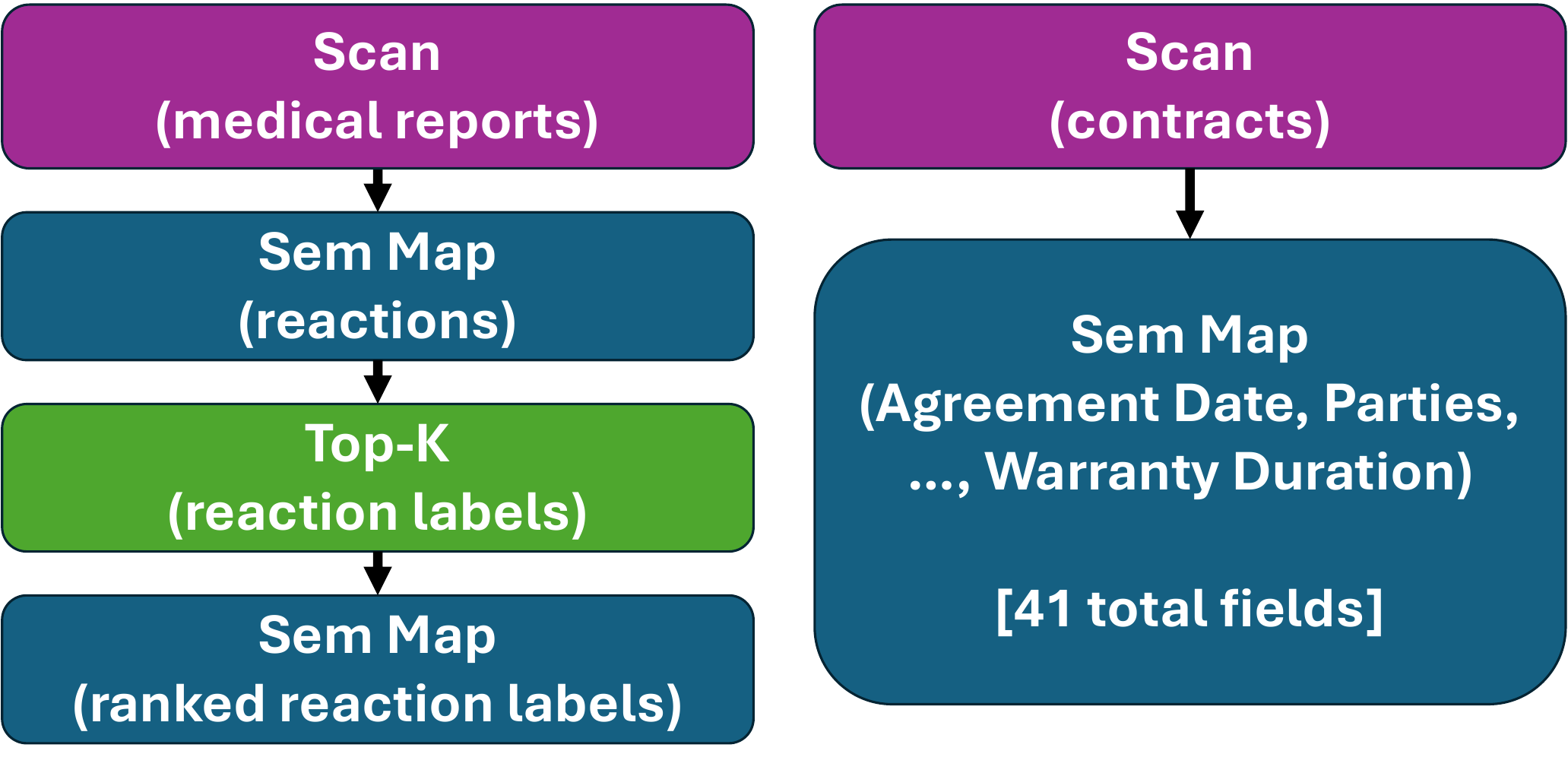}
    \caption{Query plans for BioDEX (left) and CUAD (right).}
    \Description{Query plans for the BioDEX (left) and CUAD (right) benchmarks.}
    \label{fig:biodex-cuad}
\end{figure}

\begin{figure}[t!]
    \centering
    \includegraphics[width=\linewidth]{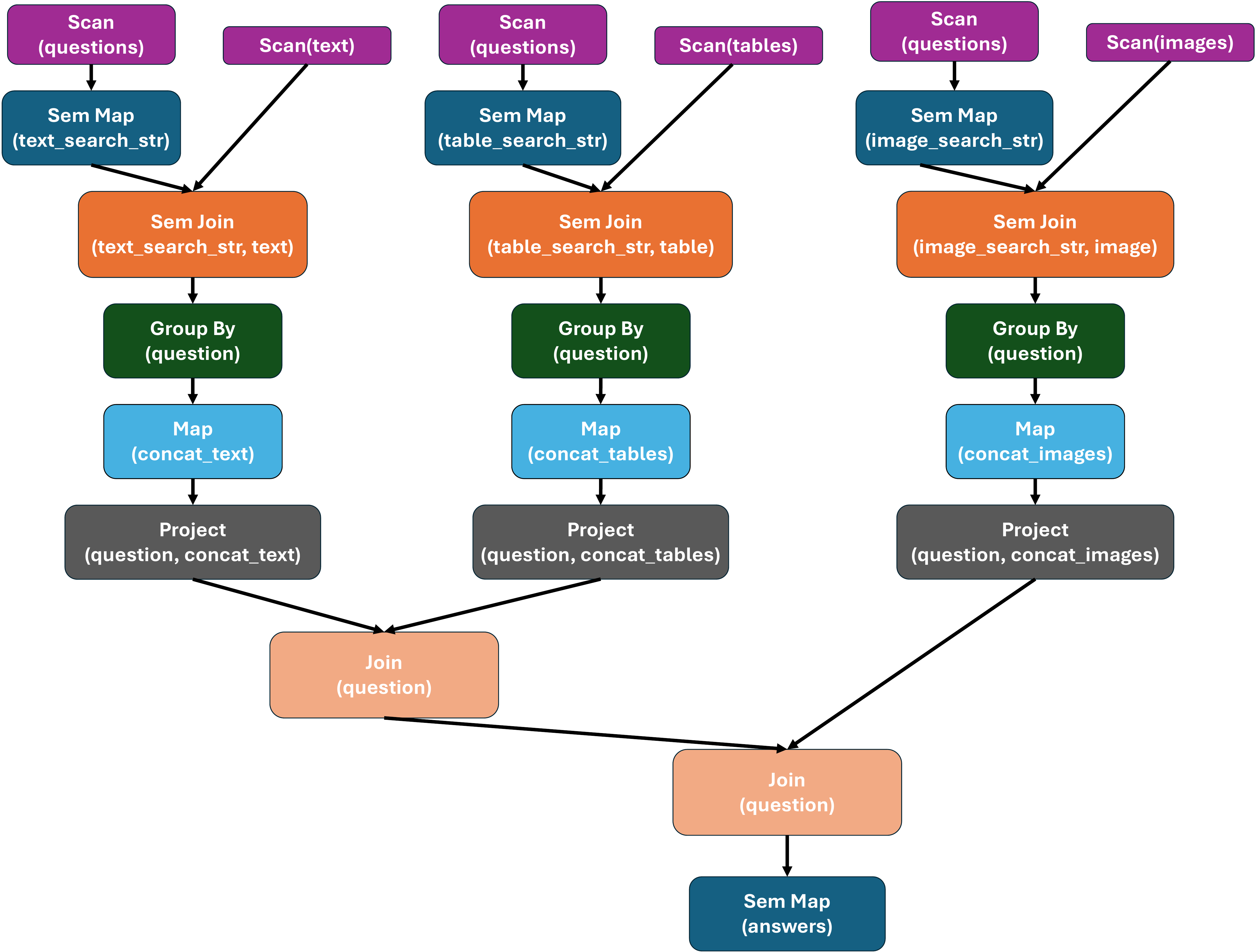}
    \caption{The query plan for the MMQA benchmark.}
    \Description{The query plan for the MMQA benchmark.}
    \label{fig:mmqa-complex}
\end{figure}

\begin{table*}[!t]
  \caption{Performance on the BioDEX, CUAD, and MMQA benchmarks for systems optimized to maximize quality. Quality is measured using RP@K for BioDEX and F1 score for CUAD and MMQA. Mean values are shown with their standard deviation.}
  \label{tab:sys-comparison}
  \begin{tabular}{c|c|>{\columncolor{magenta!20}}c|ccc|ccc}
    \toprule
     & & & & Cost (\$) & & & Time (s) & \\
     & System & Quality & Opt. & Exec. & Total & Opt. & Exec. & Total \\
    \midrule
    & DocETL & $0.193\pm0.032$ & $\$3.50\pm3.04$ & $\$3.04\pm2.51$ & $\$6.54\pm5.53$ & $427\pm130$ & $1,008\pm249$ & $1,435\pm238$ \\
    BioDEX & LOTUS & $0.216\pm0.042$ & -- & -- & $\$18.9\pm12.8$ & -- & -- & $2,348\pm1,489$ \\
    & \system{} & $\mathbf{0.261\pm0.026}$ & $\$0.18\pm0.02$ & $\$0.70\pm0.12$ & $\$0.89\pm0.11$ & $303\pm48$ & $147\pm22$ & $450\pm47$ \\
    \midrule
    & DocETL & $0.475\pm0.106$ & $\$6.04\pm2.52$ & $\$1.01\pm0.330$ & $\$7.05\pm2.63$ & $1,540\pm511$ & $280\pm128$ & $1,820\pm594$ \\
    CUAD & LOTUS & $0.234\pm0.005$ & -- & -- & $\$0.20\pm0.02$ & -- & -- & $125\pm19$\\
    & \system{} & $\mathbf{0.662\pm0.010}$ & $\$0.19\pm0.05$ & $\$0.51\pm0.01$ & $\$0.69\pm0.05$ & $318\pm61$ & $132\pm13$ & $450\pm67$ \\
    \midrule
    & GPT-4o-mini & $0.160\pm0.33$ & -- & -- & $<\$3\cdot10^{-3}$ & -- & -- & $78.0\pm4.9$ \\
    MMQA & LOTUS & $0.284\pm0.046$ & -- & -- & $\$14.3\pm5.8$ & -- & -- & $1,208\pm347$ \\
    & \system{} & $\mathbf{0.304\pm0.079}$ & $\$0.17\pm0.01$ & $\$12.9\pm10.6$ & $\$13.1\pm10.6$ & $598\pm152$ & $550\pm299$ & $1,149\pm300$ \\
    \bottomrule
  \end{tabular}
\end{table*}

\begin{table}[!t]
  \caption{Abacus' performance on BioDEX, CUAD, and MMQA when optimizing to minimize cost (top) and latency (bottom). Reduction measures how much cheaper / faster the optimized plan is relative to the max quality equivalent.
  }
  \label{tab:min-latency-min-cost}
  \begin{tabular}{c|c>{\cellcolor{magenta!20}}cc}
    \toprule
     (MinCost) & Quality & Exec. Cost (\$) & Reduction\\
    \midrule
    BioDEX & $0.21\pm0.02$ & $\$0.28\pm0.10$ & 2.50x \\
    CUAD & $0.05\pm0.02$ & $\$0.12\pm0.01$ & 4.25x \\
    MMQA & $0.31\pm0.05$ & $\$16.0\pm9.7$ & 0.81x \\
    \midrule
    \midrule
    (MinTime) & Quality & Exec. Time (s) & Reduction \\
    \midrule
    BioDEX & $0.21\pm0.03$ & $128\pm50$ & 1.15x \\
    CUAD & $0.10\pm0.05$ & $55\pm18$ & 2.4x \\
    MMQA & $0.28\pm0.07$ & $540\pm382$ & 1.02x \\
    \bottomrule
  \end{tabular}
\end{table}

\subsection{Benchmarks and Implementations.}
\textbf{Benchmarks.} We evaluate \system{} on three benchmarks for processing unstructured documents. Each input in the \textbf{BioDEX} benchmark \cite{doosterlinck2023biodex} is a document describing adverse reaction(s) experienced by a patient in response to taking a drug. In line with prior work \cite{patel2025semanticoperatorsdeclarativemodel, shankar2024docetlagenticqueryrewriting, d2024context}, we focus on the task of producing a ranked list of the adverse reactions experienced by the patient. Success on this task is measured by the rank-precision (RP) of the output rankings at a specified threshold $K$ (i.e. RP@K). Each input in the \textbf{CUAD} \cite{hendrycks2021cuad} benchmark is a legal contract. Given a set of 41 contract clauses, the task is to predict the span(s) in the contract which correspond to each clause (the ground truth for a single clause spans $\sim0.25\%$ of the document on average). Success on this task is measured by the F1-score of the clause predictions. Finally, The \textbf{MMQA} dataset \cite{talmor2021multimodalqa} contains questions involving reasoning over images, text, and/or tables. Success is measured by F1-score on question answers, as the ground truth label for every question is a list of outputs.

\textbf{Implementations.} For BioDEX, we use code from the authors of DocETL and LOTUS to evaluate their systems. Both LOTUS and DocETL compute a semantic join between each input medical document and the list of reaction labels before using a semantic map to rerank the labels. For \system{}, we implement a pipeline in Palimpzest which joins input medical documents to the most similar reaction labels using a semantic map and semantic top-k operator, before reranking the labels using a semantic map (\Cref{fig:biodex-cuad}, left). For CUAD, each framework computes the 41 output clauses with a semantic map (or semantic extract). The code for DocETL was provided by the authors, while the code for LOTUS and \system{} were implemented using a single operator (\Cref{fig:biodex-cuad}, right).

Finally, for MMQA we implement a simple baseline which asks \texttt{GPT-4o-mini} to answer each question without any relevant image, text, or table content. This represents a lower bound on expected performance. DocETL does not support image inputs so we omit it from our evaluation. For LOTUS and \system{}, we implemented a complex query plan (\Cref{fig:mmqa-complex}) that uses three semantic maps to generate search strings for querying relevant images, text, and tables. Next, each data modality is semantically joined to the questions based on their relevance to the search string. The retrieved data is then manipulated with relational group by, map, project, and join operations to obtain a dataset with one row per question, where each row contains a list of the joined images, text, and tables. Finally, the question is answered using a semantic map. (To keep the computation tractable, we limited the datasets to only include data items related to at least one question in the ground truth.)

\begin{figure*}[t!]
    \centering
    \includegraphics[width=\linewidth]{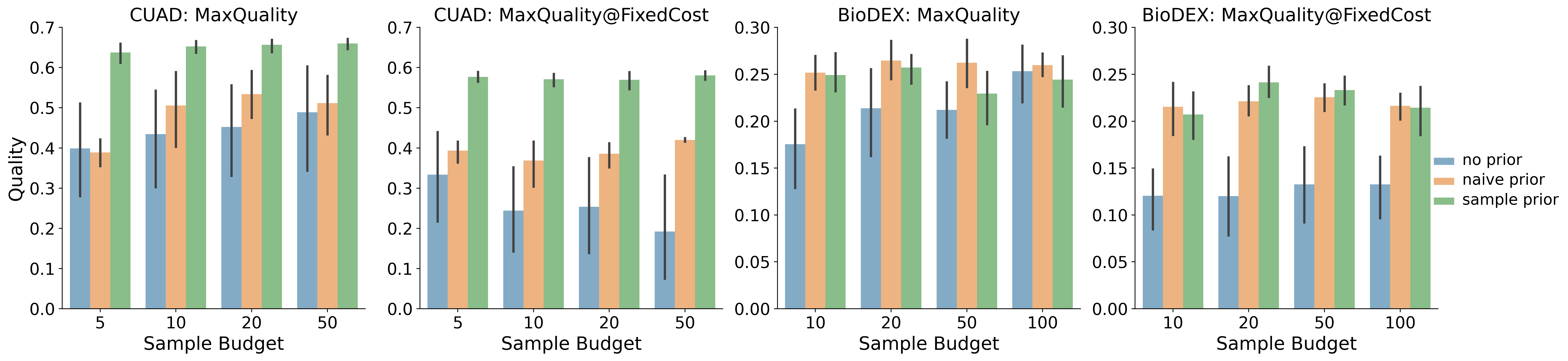}
    \caption{System output quality as a function of the sample budget when optimizing with (1) no priors, (2) naive priors computed from MMLU-Pro performance, and (3) priors computed with samples from each benchmark's train split. Overall, \system{} yields better plans in the constrained and unconstrained settings when leveraging prior beliefs on operator performance.}
    \Description{System output quality as a function of the sample budget when optimizing with (1) no priors, (2) naive priors computed from MMLU-Pro performance, and (3) priors computed with samples from each benchmark's train split. Overall, \system{} yields better plans in the constrained and unconstrained settings when leveraging prior beliefs on operator performance.}
    \label{fig:priors}
\end{figure*}

\subsection{\system{} Outperforms Prior Work}
\label{sec:abacus-vs-prior-work}
To evaluate our first experimental claim, we compare \system{} to DocETL \cite{shankar2024docetlagenticqueryrewriting} and LOTUS \cite{patel2025semanticoperatorsdeclarativemodel}. In order to maintain parity with their evaluations, we restrict each system to using \texttt{GPT-4o-mini}, \texttt{text-embedding-3-small}, and \texttt{clip-ViT-B-32}.

\textbf{Setup.} We executed each system 10 times on the BioDEX, CUAD, and MMQA benchmarks with the objective of maximizing output quality. For each of the 10 trials, we sampled different examples from the benchmarks' test datasets. These ``splits" (each containing 250 samples for BioDEX and 100 samples for CUAD and MMQA) were drawn to make the evaluation computationally tractable while still accounting for diversity in test examples. We ran each system on each split and measured the output quality, execution cost in dollars, and latency in seconds. We report the mean and standard deviation of these measurements. For DocETL and \system{}, which have distinct optimization and execution stages, we also break out the cost of optimization and the cost to execute the final optimized plan. Finally, for \system{} we used the default values of $k=6$ and $j=4$, while setting the sample budget equal to $50 \times$ the number of semantic operators in the query plan. This heuristic ensures that roughly half of the sample budget ($50 - 6 \cdot 4 = 26$) is used for exploration by the MAB algorithm.  

\textbf{Results.} The results of our evaluation are shown in \Cref{tab:sys-comparison}. Overall, \system{} is able to maximize quality better than all competing systems. On BioDEX, CUAD, and MMQA, \system{} achieves 20.3\%, 18.7\%, and 39.2\% better mean quality than the next best system, respectively. Furthermore, \system{}'s plans are on-average 12.6x cheaper and 2.7x faster than the next best system (in terms of quality). The key drivers of \system{}'s performance improvements vary across benchmarks. However, a common theme is that \system{} succeeds when it (1) has access to implementation rules (i.e. optimizations) that are cost-effective for the given task and (2) is able to obtain good estimates of these optimizations' performance, which help the MAB algorithm  search for high-quality optimizations.

For example, on the BioDEX benchmark, the Reduced-Context Generation rule works well for extracting reactions from the medical reports and re-ranking the final reaction labels (the first and second map in \Cref{fig:biodex-cuad}, respectively). This is because the optimization discards text from the medical report which is not directly related to the patient's adverse reaction(s), thus saving money on token processing while also helping the LLM focus on the most relevant data when performing the map. LOTUS optimizes semantic joins by sampling join tuples in order to learn thresholds for a cascade. However, the quality of the cascade is subject to variance and depends on how well the sampled join tuples represent the overall join. In the worst case, LOTUS produces joins which require $>100,000$ LLM calls, leading to high runtime and cost. For DocETL, we inspect the final pipelines it generates with GPT-4o and find that it often implements the query by (1) using a map to extract reactions from the underlying document, (2) joining these reactions to the reaction labels, and (3) reranking  reaction labels based on document relevance. DocETL's LLM-based optimizer implements the join using a combination of heuristics (e.g. checking that all words in the reaction label appear in the document) and thresholding based on embedding similarity between the reaction label and the extracted reactions. Similar to LOTUS, imprecision in the choice of heuristic and similarity threshold can degrade performance.

On CUAD, \system{} implements the semantic map with a Mixture-of-Agents operator in all 10 trials. The Mixture-of-Agents operator does a good job of extracting legal clauses from the documents with high precision (avg. $87.2\%$) and decent recall (avg. $53.0\%$), likely because its proposer-aggregator architecture enables the aggregator LLM to synthesize a final response which only includes the proposed text spans which are correct with high probability. By comparison, DocETL's LLM optimizer spends anywhere from 20 - 40 minutes incurring high optimization cost as it decomposes the map operation into a multi-stage pipeline. In 10 trials, we observe that DocETL rewrites the map into a pipeline with anywhere from 2 to 7 operations, which ultimately leads to large variance in its performance. Interestingly, we find that DocETL's pipelines perform best (achieving up to 63.7\% F1-score) when it composes a 3-step pipeline and perform much worse (as low as 35.3\% F1-score) on its deeper 7-step pipelines. (LOTUS does not optimize map operators so its implementation is cheap and fast but achieves low quality).

On MMQA, the implementation of the three semantic joins largely determines each system's performance. In particular, the semantic join for relevant images is the primary driver of plan cost and latency. If the joins successfully retrieve the relevant images, text, and/or tables, the final map is significantly more likely to answer the question correctly. \system{} implements the semantic image join with the EmbeddingJoin optimization on 75\% of trials. Similar to LOTUS, the optimization's performance is sensitive to the join tuples sampled during optimization. However, its conservative thresholding leads the EmbeddingJoin to (on average) invoke the LLM more often, which leads to higher average quality with slightly larger cost and latency.

In \Cref{tab:min-latency-min-cost} we examine \system{}'s ability to minimize the cost and latency of query plans on each benchmark. For BioDEX and CUAD, \system{} reduces the average cost / latency of the optimized plans by implementing each map with Reduced-Context Generation rules that process only a fraction of the document. On BioDEX, \system{} already uses instances of this rule when optimizing for quality so the cost and latency savings are smaller than for CUAD. On CUAD, these savings come with a larger trade-off in quality, because text related to the 41 legal clauses are more evenly dispersed through the contract. Finally, \system{} struggles to minimize latency and cost on MMQA because 75\% of the quality maximizing plans already use the EmbeddingJoin optimization, leaving little room for \system{} to improve. Furthermore, because the EmbeddingJoin exhibits high variance in its quality, cost, and latency, when averaged over 10 trials we can achieve results where the minimum cost / latency plans achieve higher cost / latency than their quality maximizing counterparts. In the future--if given better, lower variance optimizations--we expect \system{} could improve its cost and latency minimizing results in \Cref{tab:min-latency-min-cost}.

\subsection{Performance Improves with Better Priors}
\label{sec:priors}
For our second experimental claim, we ran \system{} on CUAD and BioDEX with and without prior beliefs while also varying the sample budget. We omitted MMQA from our evaluation as the number of physical operators for semantic joins is small enough to sample exhaustively. We examined maximizing quality with and without a cost constraint. The cost constraints for CUAD and BioDEX were set equal to the 25th percentile of plan execution costs we observed in the unconstrained setting, thus making them non-trivial to satisfy. We aimed to show that \system{} could leverage prior beliefs to identify more optimal plans with fewer samples.

For each benchmark, we used two sets of prior belief(s). The first ``naive" prior estimated each operator's quality as an average of its model(s') performance on the MMLU-Pro benchmark \cite{wang2024mmlu}. It also estimated the cost of each operator by averaging its per-token input and output costs. This prior is cheap to compute and can be done offline, but lacks fidelity in the accuracy of its estimates. The second ``sample-based" prior estimated operator performance by running each operator on 5 samples from the train split of the respective dataset. This prior is more expensive to compute and must be done online, but has higher fidelity in its estimates.

The results of our evaluation are shown in \Cref{fig:priors}. Overall, we observe that \system{} produces plans with higher quality when provided with prior beliefs. In the unconstrained setting, plans optimized with prior beliefs perform up to 1.60x and 1.43x better (at a fixed sample budget) than those optimized without priors on CUAD and BioDEX, respectively. This gap is even greater in the constrained optimization setting, where plans optimized with prior beliefs perform up to 3.02x and 2.01x better (at a fixed sample budget) than those optimized without priors on CUAD and BioDEX, respectively. This latter result comes from the fact that identifying a good Pareto frontier of operators is more difficult than identifying a single best operator, thus having a prior belief over the entire frontier provides greater benefit relative to sampling without priors.

Finally, due to (1) the higher average operator costs in this experiment and (2) the small sample budget sizes, average plan costs and latencies are higher than in \Cref{tab:sys-comparison}. The average plan costs and latencies in the constrained and unconstrained settings, respectively, are ($\$4.42$, 240s) and ($\$1.15$, 211s) for BioDEX and ($\$6.20$, 213s) and ($\$1.83$, 214s) for CUAD. Still, 91.6\% and 93.3\% of the constrained BioDEX and CUAD plans satisfy the given constraints, suggesting the optimizer was willing to accept higher costs and runtimes in the pursuit of better plan quality.

\subsection{\system{} Leverages Relaxed Constraints}
\begin{figure}
    \centering
    \includegraphics[width=\linewidth]{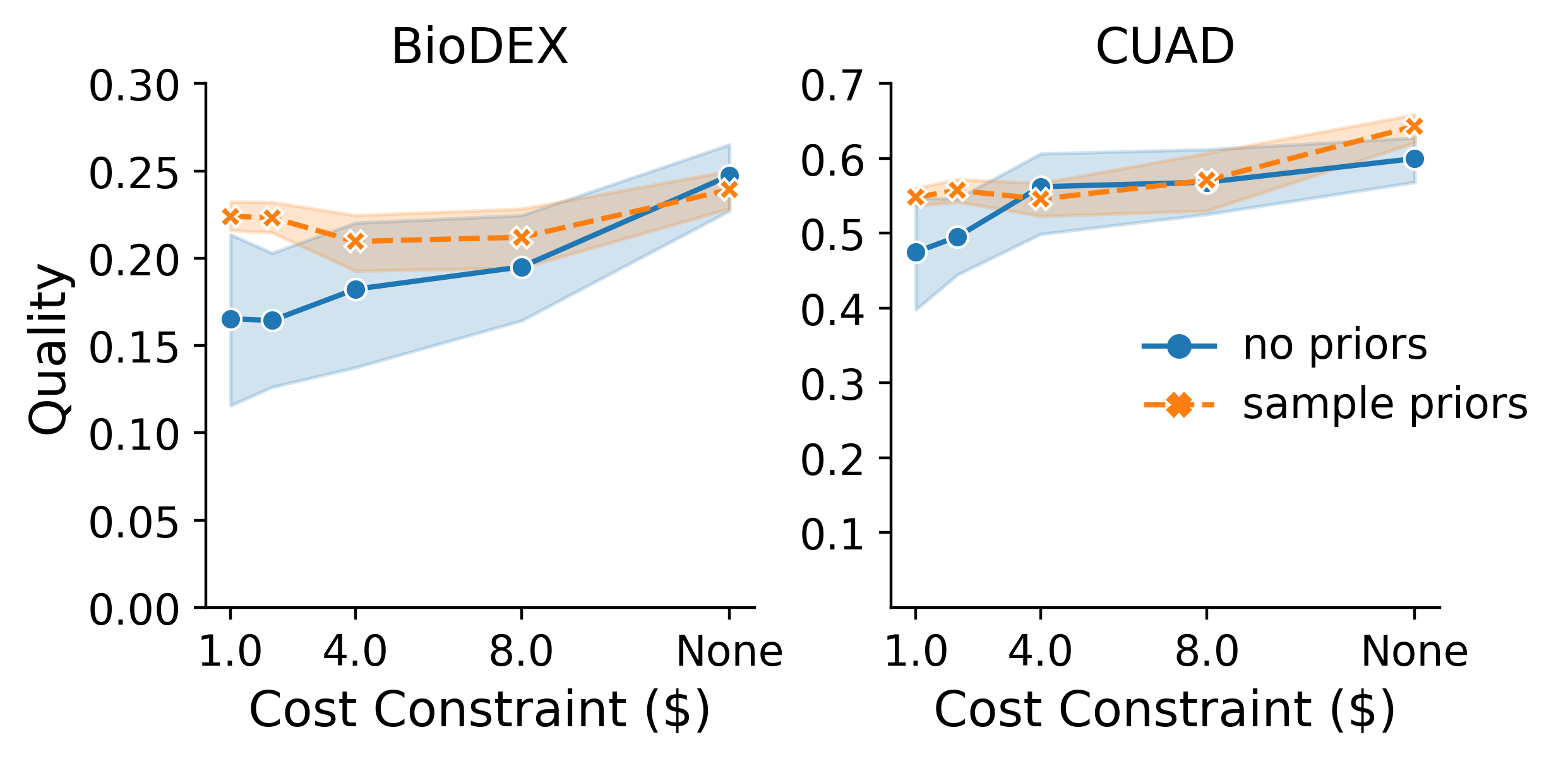}
    \caption{The performance of plans optimized for max quality subject to a cost constraint. Plan performance improves as constraints are relaxed. Prior beliefs help \system{} maintain better performance even with tight cost constraints.}
    \Description{The performance of plans optimized for max quality subject to a cost constraint. Plan performance improves as constraints are relaxed. Prior beliefs help \system{} maintain better performance even with tight cost constraints.}
    \label{fig:cost-thresholds}
\end{figure}
\begin{figure}
    \centering
    \includegraphics[width=\linewidth]{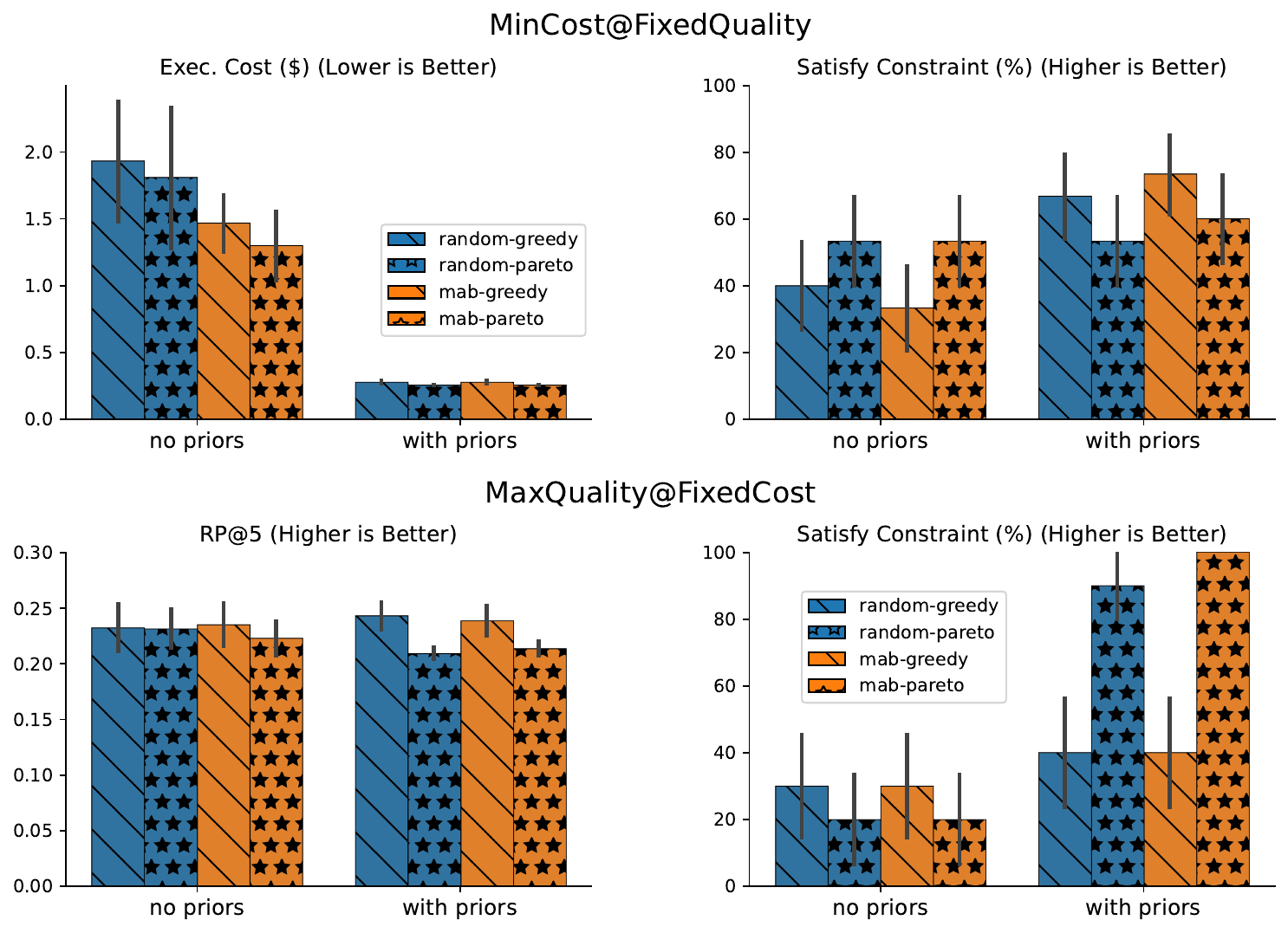}
    \caption{Ablation study isolating the effects of the Pareto-Cascades and MAB sampling algorithms and prior beliefs.}
    \Description{Ablation study isolating the effects of the Pareto-Cascades and MAB sampling algorithms and prior beliefs on BioDEX. The objective was to minimize cost at a quality constraint of 80\% of \system{}'s mean performance in \Cref{tab:sys-comparison}.}
    \label{fig:ablation}
\end{figure}
For our third experimental claim, we used \system{} to optimize plans for BioDEX and CUAD with the objective of maximizing quality subject to a cost constraint. We varied the cost constraint from unconstrained optimization down to \$1, which is 11.8\% and 16.2\% of the median cost of an unconstrained plan on BioDEX and CUAD, respectively. Tightening the cost constraint limits the space of systems available to \system{}. Thus, our goal was to demonstrate that \system{} responds to looser constraints by identifying more optimal plans, or vice-versa, that \system{} responds to tighter constraints with non-trivial system implementations.

For each cost constraint, we used \system{} to optimize plans for maximum quality with 10 different test splits of the BioDEX and CUAD datasets. We used the same sample budget at each cost constraint and optimized with and without prior beliefs. The results of our evaluation are shown in \Cref{fig:cost-thresholds}. For optimization without prior beliefs, \system{} is generally able to identify plans which achieve better quality as the cost constraint is relaxed. Furthermore, \system{} still identifies plans which achieve non-trivial performance when optimizing under tight constraints.

When optimizing with priors, we see a smaller degradation in performance as the constraint is tightened. For example, without priors, performance on BioDEX decreases by 45.6\% from having no constraint to a constraint of \$1. However, with priors it only decreases by 12.5\% at its lowest point with a constraint of \$4. This is due to the fact that prior beliefs on operator performance help guide \system{}'s MAB sampling algorithm to prioritize operators which lie on the entire Pareto frontier of the cost vs. quality trade-off.

\subsection{Ablation Study}
Finally, we ran an ablation study to isolate the benefits of Pareto-Cascades, our MAB sampling algorithm, and prior beliefs. We ran \system{} on the BioDEX benchmark with the default values of $k$ and $j$ and a sample budget of 150 samples. We executed two policies: minimizing cost with a lower bound on quality (\Cref{fig:ablation}, top row) and maximizing quality with an upper bound on cost (\Cref{fig:ablation}, bottom row). To make the optimization challenging, we set the quality and cost constraints equal to 80\% and 50\% of \system{}'s mean quality and cost in \Cref{tab:sys-comparison}, respectively.

When running \system{} without Pareto-Cascades, we replaced it with a modified traditional Cascades algorithm that greedily selects the optimal subplan for each group which does not violate the constraint. This algorithm represents a greedy scheme in which the plan is built myopically without considering how subsequent operators may need to satisfy the constraint. When running \system{} without the MAB sampling algorithm, we replaced it with a random sampling algorithm that sampled $k$ physical operators and $j = B / k$ inputs. In theory, this algorithm should suffer from an inability to use some of its sample budget to search for new operators, which is a key benefit of the MAB approach.

The results of our ablation are shown in \Cref{fig:ablation}. First, on-average, prior beliefs help \system{} optimize the objective 3.5x better and satisfy the constraint 1.8x more often. Second, we can see the benefits of the MAB algorithm in the cost minimization objective, where \system{} with the MAB algorithm minimizes cost by 1.4x and 1.3x for Pareto-Cascades and greedy optimized plans, respectively. These benefits do not show up as clearly in the quality maximizing objective; this is in part because estimating quality is more difficult and may require more samples for this query. Third, on-average, the Pareto-Cascades algorithm helps identify plans which satisfy constraints 1.2x more often.

%% file: sections/limitations.tex
\section{Limitations and Future Work}
\label{sec:limitations}

\textbf{Modeling Operator Dependencies.} One limitation of \system{}'s cost model is that it treats operator performance as being independent of the other operators in the plan. This assumption enables \system{} to produce cost estimates for logical plans which it has not sampled. However, it can also lead to errors in estimation. In future work we plan to explore using techniques from Bayesian optimization \cite{garnett2023bayesian}, which have recently been applied in similar declarative programming frameworks \cite{khattab2023dspycompilingdeclarativelanguage}.

\textbf{Sequential MAB Sampling.} In \system{}' implementation of \Cref{algo:mab-sampling}, each operator frontier is updated in sequence in order to determine the new highest quality operator(s) on the frontier. The sample output(s) generated by the highest quality operator(s) are then used as input to the next operator frontier in Line 7 of \Cref{algo:abacus}. A downside of this implementation is that the optimization algorithm is slowed  by the sequential processing of operator frontiers (even though operators within a frontier may process sample inputs in parallel). In the future, we will explore pipelining  execution of operator frontiers to decrease  optimization overhead. 

%% file: sections/related_work.tex
\section{Related Work}

\textbf{Optimizing Semantic Operator Systems.} Recent work has investigated the optimization of semantic operator systems \cite{liu2025palimpzest, patel2025semanticoperatorsdeclarativemodel, shankar2024docetlagenticqueryrewriting, aryn2025, lu2025vectraflow, urban2023caesuralanguagemodelsmultimodal, evadb2023, lin2024accurateefficientdocumentanalytics, galois2025satriani}. Prior to \system{}, Palimpzest \cite{liu2025palimpzest} used sampling and heuristics to estimate and optimize semantic operator systems. LOTUS \cite{patel2025semanticoperatorsdeclarativemodel} optimizes semantic join, filter, group-by, and top-k operators by offloading data processing from an expensive ``gold algorithm" to a cheaper proxy method while providing statistical guarantees on quality with respect to the gold algorithm. DocETL \cite{shankar2024docetlagenticqueryrewriting} uses LLMs to apply (and validate) query rewrites to data processing pipelines. In contrast to \system{}, LOTUS and DocETL only optimize for system quality and do not satisfy explicit constraints on system cost or latency.



Similar to DocETL, Aryn \cite{aryn2025} uses an LLM to apply rewrites to query plans, however it focuses more on using a human-in-the-loop to validate the plans it generates. VectraFlow \cite{lu2025vectraflow} built a stream processing engine with support for vector data and vector-based operations. Galois \cite{galois2025satriani} introduced new logical and physical optimizations for answering queries with LLM-based operators.

Early work on semantic operators focused on adding machine learning classifiers to data systems for tasks such as image classification, object detection, sentiment analysis, and more \cite{Anderson2018PhysicalRP, kang2017noscope, kang2019blazeit, russo2023inquest, kang2021abae, kossmann2023etl}. Caesura \cite{urban2023caesuralanguagemodelsmultimodal}, EVA \cite{evadb2023}, and ZenDB \cite{lin2024accurateefficientdocumentanalytics} integrated semantic operators into systems which support SQL queries over multi-modal, video, and text data, respectively. In general, these systems support narrow optimizations over semantic operators, rather than building a new general-purpose optimizer for them.


\textbf{Optimizing More General AI Systems.}
There is also a large body of work on building AI systems that go beyond using semantic operators. We focus our discussion on frameworks which treat the optimization of these systems as a primary challenge. DSPy \cite{khattab2023dspycompilingdeclarativelanguage, opsahlong2024optimizinginstructionsdemonstrationsmultistage, soylu2024finetuningpromptoptimizationgreat} enabled users to construct and optimize ``language model programs", i.e. workflows composed of modular operators which can be optimized in a declarative manner. The main levers of optimization included prompt optimization, parameter optimization, and model finetuning. More recently, AFlow \cite{zhang2025aflow, hong2024metagpt}, Archon \cite{saadfalcon2024archonarchitecturesearchframework}, and ADAS \cite{hu2025automated} all explored automatically constructing AI systems for a given workload. Each of these systems searches over a space of computation graphs and operator implementations, using Monte Carlo Tree Search, Bayesian optimization, and LLM-guided search, respectively. In contrast to these works, \system{} focuses solely on declarative optimization of semantic operator systems.

\textbf{Relational Query Optimization.}
There is a long and rich literature on query optimization in relational database systems \cite{selinger, starburst, volcano, Graefe1995TheCF, Neumann2017TheCS, neumann2018adaptive}. From this line of work, \system{} most closely resembles a Cascades optimizer \cite{Graefe1995TheCF}. There are two key challenges which make optimizing semantic operator systems different from optimizing relational queries. First, the quality of a semantic operator is not guaranteed to be perfect. Thus, \system{} must be able to estimate the quality of an operator, possibly without the use of precomputed statistics. Second, in order to support constrained optimization \system{} cannot rely on the principle of optimality to prune sub-plans during its plan search. This necessitates \system{}'s use of a new dynamic programming algorithm which maintains the Pareto frontier of physical plans for every subplan in its search.

%% file: sections/conclusion.tex
\section{Conclusion}
We present \system{}, an extensible, cost-based optimizer for semantic operator systems which optimizes their quality, cost, and latency. \system{} modifies traditional multi-armed bandit and Cascades algorithms to overcome challenges in estimating the performance of semantic operators while supporting constrained optimization. We evaluate \system{} and its core algorithmic contributions on a diverse set of benchmarks. We demonstrate that its plans perform better than those produced by recent work, and its algorithmic contributions help decrease the number of samples required to achieve good performance and help satisfy optimization constraints.
